\documentclass[11pt,a4paper]{amsart}
\usepackage[margin=20mm]{geometry}
\usepackage[numbers]{natbib}
\usepackage{hyperref}

\usepackage{amsaddr}

\usepackage{amssymb,amsmath}
\usepackage{amsthm}
\usepackage{bbm}
\usepackage{stmaryrd}
\usepackage[usenames,dvipsnames]{xcolor}
\usepackage{color}
\usepackage{enumitem}
\usepackage{multirow}

\input xy
\xyoption{all} 

\numberwithin{equation}{section}

\newtheorem{theorem}{Theorem}[section]

\newtheorem{proposition}[theorem]{Proposition}

\theoremstyle{definition}
\newtheorem{definition}[theorem]{Definition}
\newtheorem{remark}[theorem]{Remark}
\newtheorem{example}[theorem]{Example}

\newcommand{\Id}{\mathbbmss{1}}

\newcommand{\rmi}{ \textnormal{i}}

\newcommand{\rme}{\textnormal{e}}

\DeclareMathOperator{\Vect}{Vect}

\DeclareMathOperator{\Span}{Span}

\DeclareMathOperator{\Der}{Der}

\font\black=cmbx10 \font\sblack=cmbx7 \font\ssblack=cmbx5 \font\blackital=cmmib10  \skewchar\blackital='177
\font\sblackital=cmmib7 \skewchar\sblackital='177 \font\ssblackital=cmmib5 \skewchar\ssblackital='177
\font\sanss=cmss10 \font\ssanss=cmss8 
\font\sssanss=cmss8 scaled 600 \font\blackboard=msbm10 \font\sblackboard=msbm7 \font\ssblackboard=msbm5
\font\caligr=eusm10 \font\scaligr=eusm7 \font\sscaligr=eusm5  \font\fraktur=eufm10
\font\sfraktur=eufm7 \font\ssfraktur=eufm5 
\font\bsymb=cmsy10 scaled\magstep2
\def\all#1{\setbox0=\hbox{\lower1.5pt\hbox{\bsymb
       \char"38}}\setbox1=\hbox{$_{#1}$} \box0\lower2pt\box1\;}
\def\exi#1{\setbox0=\hbox{\lower1.5pt\hbox{\bsymb \char"39}}
       \setbox1=\hbox{$_{#1}$} \box0\lower2pt\box1\;}

\def\tx#1{{\fam0\relax#1}}

\newfam\bifam
\textfont\bifam=\blackital \scriptfont\bifam=\sblackital \scriptscriptfont\bifam=\ssblackital

\newfam\blfam
\textfont\blfam=\black \scriptfont\blfam=\sblack \scriptscriptfont\blfam=\ssblack

\newfam\bbfam
\textfont\bbfam=\blackboard \scriptfont\bbfam=\sblackboard \scriptscriptfont\bbfam=\ssblackboard

\newfam\ssfam
\textfont\ssfam=\sanss \scriptfont\ssfam=\ssanss \scriptscriptfont\ssfam=\sssanss
\def\sss#1{{\fam\ssfam\relax#1}}

\newfam\clfam
\textfont\clfam=\caligr \scriptfont\clfam=\scaligr \scriptscriptfont\clfam=\sscaligr

\newfam\frfam
\textfont\frfam=\fraktur \scriptfont\frfam=\sfraktur \scriptscriptfont\frfam=\ssfraktur

\def\hpb#1{\setbox0=\hbox{${#1}$}
    \copy0 \kern-\wd0 \kern.2pt \box0}
\def\vpb#1{\setbox0=\hbox{${#1}$}
    \copy0 \kern-\wd0 \raise.08pt \box0}

\def\pmb#1{\setbox0\hbox{${#1}$} \copy0 \kern-\wd0 \kern.2pt \box0}
\def\pmbb#1{\setbox0\hbox{${#1}$} \copy0 \kern-\wd0
      \kern.2pt \copy0 \kern-\wd0 \kern.2pt \box0}
\def\pmbbb#1{\setbox0\hbox{${#1}$} \copy0 \kern-\wd0
      \kern.2pt \copy0 \kern-\wd0 \kern.2pt
    \copy0 \kern-\wd0 \kern.2pt \box0}
\def\pmxb#1{\setbox0\hbox{${#1}$} \copy0 \kern-\wd0
      \kern.2pt \copy0 \kern-\wd0 \kern.2pt
      \copy0 \kern-\wd0 \kern.2pt \copy0 \kern-\wd0 \kern.2pt \box0}
\def\pmxbb#1{\setbox0\hbox{${#1}$} \copy0 \kern-\wd0 \kern.2pt
      \copy0 \kern-\wd0 \kern.2pt
      \copy0 \kern-\wd0 \kern.2pt \copy0 \kern-\wd0 \kern.2pt
      \copy0 \kern-\wd0 \kern.2pt \box0}

\mathchardef\za="710B  
\mathchardef\zb="710C  
\mathchardef\zg="710D  
\mathchardef\zd="710E  
\mathchardef\zve="710F 
\mathchardef\zz="7110  
\mathchardef\zh="7111  
\mathchardef\zvy="7112 
\mathchardef\zi="7113  
\mathchardef\zk="7114  
\mathchardef\zl="7115  
\mathchardef\zm="7116  
\mathchardef\zn="7117  
\mathchardef\zx="7118  
\mathchardef\zp="7119  
\mathchardef\zr="711A  
\mathchardef\zs="711B  
\mathchardef\zt="711C  
\mathchardef\zu="711D  
\mathchardef\zvf="711E 
\mathchardef\zq="711F  
\mathchardef\zc="7120  
\mathchardef\zw="7121  
\mathchardef\ze="7122  
\mathchardef\zy="7123  
\mathchardef\zf="7124  
\mathchardef\zvr="7125 
\mathchardef\zvs="7126 
\mathchardef\zf="7127  
\mathchardef\zG="7000  
\mathchardef\zD="7001  
\mathchardef\zY="7002  
\mathchardef\zL="7003  
\mathchardef\zX="7004  
\mathchardef\zP="7005  
\mathchardef\zS="7006  
\mathchardef\zU="7007  
\mathchardef\zF="7008  
\mathchardef\zW="700A  
\mathchardef\zC="7009  

\newcommand{\be}{\begin{equation}}
\newcommand{\ee}{\end{equation}}

\newcommand{\bea}{\begin{eqnarray}}
\newcommand{\eea}{\end{eqnarray}}
\def\*{{\textstyle *}}
\newcommand{\R}{{\mathbb R}}

\newcommand{\C}{{\mathbb C}}
\newcommand{\Z}{{\mathbb Z}}

\newcommand{\s}{{\textstyle *}}

\def\Sec{\sss{Sec}}
\def\Vect{\sss{Vect}}

\def\sT{{\sss T}}

\def\xi{\tx{i}}

\def\s*{{\scriptstyle *}}

\def\cO{\mathcal{O}}

\newcommand{\beas}{\begin{eqnarray*}}
\newcommand{\eeas}{\end{eqnarray*}}

\title{Foundations of Noncommutative Carrollian Geometry via Lie-Rinehart Pairs}
\author{Andrew James Bruce}  
   \email{andrewjamesbruce@googlemail.com}
  
   \date{\today}
 
\begin{document}
\begin{abstract}
Carrollian manifolds offer an intrinsic geometric framework for the physics in the ultra-relativistic limit. The recently introduced Carrollian Lie algebroids are generalised to the setting of $\rho$-commutative geometry (also known as almost commutative geometry), where the underlying algebras commute up to a numerical factor.  Via  $\rho$-Lie-Rinehart pairs, it is shown that the foundational tenets of Carrollian geometry have analogous statements in the almost commutative world. We explicitly build two toy examples: we equip the extended quantum plane and the noncommutative $2$-torus with Carrollian structures. This opens up the rigorous study of noncommutative Carrollian geometry via almost commutative geometry.    \par
\end{abstract}

 \maketitle

\setcounter{tocdepth}{2}
 \tableofcontents
 
\section{Introduction and Background} 
\subsection{Introduction} On the smallest of scales, it is generally believed that spacetime will radically deviate from its large-scale smooth manifold structure.  Various approaches to quantum gravity, such as string theory and loop quantum gravity, suggest that near the Planck scale, spacetime will be noncommutative in nature; hence the interest in noncommutative geometry (see Connes \cite{Connes:1994}, Madore \cite{Madore:1999}, and Manin \cite{Manin:1991}, for example). Perturbative quantum general relativity shows that noncommutative geometry is largely unavoidable in quantum gravity (see Fröb, Much \& Papadopoulos \cite{Fröb:2023}).  Typically, one is interested in constructing `relativistic quantum gravity', that is, a quantum theory that in an appropriate limit reduces to general relativity, possibly with small corrections. For a comprehensive overview of quantum gravity, the reader may consult \cite{Bambi:2024}. \par 
Over the past decade, there has been a resurgence of interest in the ultra-relativistic limit, where the speed of light approaches zero. The original works on what is now known as \emph{Carrollian physics} are those of Lévy-Leblond \cite{Lévy-Leblond:1965}, Sen Gupta \cite{SenGupta:1966}, and Henneaux \cite{Henneaux:1979}. In part, the revival in Carrollian physics is due to flat space holography, where the space boundary is defined by the endpoints of light rays, so the null direction (for an introduction, the reader may consult Nguyen \cite{Nguyen:2025}). The theory posits that a bulk gravity theory is dual to a Carrollian field theory on the boundary. For a review of the various facets of Carrollian physics, the reader may consult Bagchi et al. \cite{Bagchi:2025}.  Carrollian manifolds, so the intrinsic approach to the geometry of the ultra-relativistic limit was initiated by Duval et al. \cite{Duval:2014,Duval:2014a,Duval:2014b}, as a smooth manifold equipped with a degenerate metric whose kernel is rank-$1$. The kernel is usually assumed to have a chosen generator known as the Carroll vector field. For a review of Carrollian geometry, the reader may consult Ciambelli \& Jai-akson \cite{Ciambelli:2025}.   Carrollian geometry is the natural language for null hypersurfaces. The author recently introduced the notion of a Carrollian Lie algebroid as a framework to cope with singular Carroll distributions (see \cite{Bruce:2025}). \par
There have been almost no works on bringing together the two extremes of gravitational physics by considering `quantum/noncommutative Carrollian geometry'; see Ballesteros et al. \cite{Ballesteros:2020,Ballesteros:2020b,Ballesteros:2023}, Bose et al. \cite{Bose:2025}, and Trześniewski \cite{Trześniewski:2024}, for work in this direction, which are largely based on Inönü--Wigner contractions of $\kappa$-deformed spacetimes and their symmetries. Moreover, in $2+1$ dimensions, the Carroll group has a two-fold central extension which makes the coordinates noncommuting in the presence of a magnetic field; this is similar but distinct from the case of the $2+1$-dimensional Galilean group, see \cite{Marsot:2022,Zeng:2025}.  It is plausible that any insight into Carrollian gravity would shed light on a full theory of quantum gravity. In this paper, we make an initial study of noncommutative Carrollian geometry via graded almost commutative geometry by defining Carrollian structures on particular Lie-Rinehart pairs\footnote{The term \emph{Lie-Rinehart algebras} is also common in the literature. Historically, \emph{Lie-Cartan pairs} and \emph{Lie pseudoalgebras} have also been used.}, i.e., on the algebraic relatives of Lie algebroids.   \par 
Almost commutative algebras, or $\rho$-commutative algebras, are graded associative algebras in which the elements commute up to a numerical factor, e.g., $x y = q\, y x $, where $q$ is a non-zero element of the ground field $\mathbb{K}$. In particular, the algebras are graded by a finite abelian group $G$, and the commutation rules are controlled by a map $\rho : G \times G \rightarrow \mathbb{K}^\times$, known as the commutation factor.  Such algebras have been studied since the 1980s following the pioneering work of Rittenberg and Wyler \cite{Rittenberg:1978} and Scheunert \cite{Scheunert:1979} on generalised or colour Lie algebras. It was argued by Bongaarts and Pijls \cite{Bongaarts:1994} that almost commutative algebras offer a convenient framework to develop noncommutative differential geometry, as many of the technical difficulties with more general noncommutative algebras are absent. Using almost commutative algebras, one can largely mimic classical differential geometry following the derivation-based approach to noncommutative geometry of Dubois-Violette \cite{Dubois-Violette:1988}. For example, $\rho$-derivations of an almost commutative algebra form a module over the whole of the algebra - analogous statements for more general noncommutative algebras do not hold. Developing noncommutative geometry using almost commutative algebras allows one to largely avoid the machinery of $C^*$-algebras, von Neumann algebras, Hopf algebras, and related concepts. Many of the standard examples of noncommutative geometries are almost commutative; for example, we have supermanifolds and $\Z_2^n$-manifolds (see \cite{Covolo:2016}), quantum (super)planes, noncommutative tori, matrix algebras, and quaternionic algebras.  Thus, while concentrating on almost commutative algebras initially seems very restrictive, we do have many examples.   \par 
Lie-Rinehart pairs are the algebraic analogue of a Lie algebroid, and as the algebra need not be commutative, they offer a framework to define ``noncommutative Lie algebroids''. We work with a refined version that we refer to as \emph{$\rho$-Lie-Rinehart pairs}, where the grading of the almost commutative algebra, $\rho$-Lie bracket, and anchor map are all compatible (see Definition \ref{def:rhoLieRinPai} and \cite[Definition 3.1]{Ngakeu:2017}). In doing so, we are able to define $\rho$-connections following the notion of a Lie algebroid connection; importantly, the curvature and torsion are (graded almost commutative) tensor-objects.  From there, we define Carrollian structures on $\rho$-Lie-Rinehart pairs understood as  ``almost commutative Carrollian manifolds'' (see Definition \ref{def:CarrhoLieRinPai}).\par 
The work presented here is foundational and proposes an approach to noncommutative Carrollian geometry. However, further examples are needed to assess the full utility of Carrollian $\rho$-Lie-Rinehart pairs in physics. Nonetheless, novel concepts and new ideas in noncommutative geometry are given.

\medskip

\noindent \textbf{Outline of the paper.} We continue this section with Subsection \ref{subsec:AlComAlg} in which we recall $G$-graded algebras, almost commutative algebras, and their (graded) derivations. The main bulk of the paper is contained in Section \ref{sec:LieRinPaiCarr}. In Subsection \ref{subsec:rhoLieRinPai} we define and make an initial study of $\rho$-Lie-Rinehart pairs understood as `almost commutative Lie algebroids'. In subsection \ref{subsec:ConrhoLieRinPai}, we show that connections on  $\rho$-Lie-Rinehart pairs can be defined and that the core elements of Lie algebroid connections generalises to this setting. In Subsection \ref{subsec:CarrrhoLieRinPair}, we define Carrollian $\rho$-Lie-Rinehart pairs as graded almost commutative generalisations of Carrollian manifolds and present their core properties.  Two toy examples are presented in Subsections \ref{subsec:ToyExa}. Specifically, we equip Manin's extended quantum plane and a canonical $\rho$-Lie-Rinehart pair over the noncommutative $2$-torus, with selected Carrollian structures and compatible connections. 

\medskip 
\noindent \textbf{Conventions.}  We will take the ground field to be  $\mathbb{K} = \R$ or $\C$. All associative algebras will be taken to be unital. When dealing with graded algebras, we will take the liberty of writing expressions for homogeneous elements and understand that extension to general elements is via linearity.  
\medskip 

Heuristically, on Minkowski space-time, in the ultra-relativistic limit light cones collapse to lines and  massive particles become `frozen' in space, i.e., there is no spatial propagation. This leads to the strange notion of Carrollian causality, where  we have an absolute past and future of a massive particle defined by the spatial point and the null direction.  The challenge in the almost commutative setting is to rephrase this in an algebraic language. The reader may consult the following table, which connects the almost commutative notions found in this paper to the classical notions.
\smallskip 

\renewcommand{\arraystretch}{1.5}
\begin{tabular}[h]{ |p{200pt}||p{240pt} | }
 \hline
 \multicolumn{2}{|c|}{\textbf{Dictionary}} \\
 \hline
 Almost Commutative & Commutative \\
 \hline
 $\rho$-commutative algebra $\mathcal{A}$   & algebra of smooth functions $C^\infty(M)$   \\
 $\rho$-Lie algebroid $(\mathcal{A}, \mathfrak{g})$&  Lie algebroid $\pi :A \rightarrow M $, e.g, $\sT M$ \\
 $\mathcal{A}$-module $\mathfrak{g}$, e.g. $\rho\Der(\mathcal{A})$ & $C^\infty(M)$-module $\Sec(A)$, e.g., $\Vect(M)$ \\
 anchor map $\mathsf{a} : \mathfrak{g} \rightarrow \rho\Der(\mathcal{A})$    & anchor map $\mathsf{a} : A \rightarrow \sT M$\\
 degenerate metric $\mathcal{G} : \mathfrak{g} \times \mathfrak{g} \rightarrow \mathcal{A}$& degenerate metric $g :\Sec(A)\times \Sec(A) \rightarrow C^\infty(M)$   \\
 free cyclic submodule $\mathfrak{l} := \ker(\mathcal{G})\subset \mathfrak{g}$& trivial line subbundle $L := \ker(g) \subset A$  \\
 degree zero free generator of $\mathfrak{l}$& frame of $\Sec(L)$ \\
 Carroll distribution $\mathcal{C} := \mathsf{a}(\mathfrak{l}) \subset \rho\Der(\mathcal{A}) $& Carroll distribution $\mathcal{C} := \mathsf{a}(L) \subset \sT M $\\
 \hline
\end{tabular}
\subsection{Almost Commutative Algebras and Related Notions}\label{subsec:AlComAlg} We will draw heavily on  \cite{Bongaarts:1994,Bruce:2020,deGoursac:2012,Ngakeu:2017} and references therein. Proofs of statements in this section follow via direct computation or can be found in the literature, and so will be omitted. Let $G$ be a discrete abelian group, which we will write additively. Recall that a $G$-graded algebra is an associative unital algebra (over a field $\mathbb{K}$) such that:
\begin{enumerate}[itemsep=0.5em]
    \item $\mathcal{A}= \displaystyle \bigoplus_{a \in G}\mathcal{A}_a$\,; and
    \item $A_a \cdot A_b  \subset A_{a+b}$.
\end{enumerate}
If $f \in \mathcal{A}$ is homogeneous, we will write the $G$-degree as $|f| \in G$.  Inhomogeneous elements of $\mathcal{A}$ are sums of homogeneous elements. Assume that we are given a map $\rho : G \times G \rightarrow \mathbb{K}$ that satisfies
\begin{equation}
    \rho(a,b) =  \rho(b,a)^{-1}\,, \qquad \rho(a+b,c) = \rho(a,c)\rho(b,c)\,,
\end{equation}
for all $a,b, c \in G$.  One can directly deduce, using the additive nature of the group $G$, the following derived identities: 
\begin{equation}
\rho(a,b) \neq 0\,, \quad \rho(0,b) = 1\,, \quad \rho(c,c) = \pm 1\, , \quad \rho(a, b+c) = \rho(a, b) \rho(a,c)\,.
\end{equation}
We will refer to $\rho$ as the \emph{commutation factor}. For homogeneous elements $f,g \in \mathcal{A}$, the $\rho$-commutator is defined as 
\begin{equation}\label{eqn:RhoCom}
[f,g ] :=  fg - \rho(|f|, |g|) \, g f\,.
\end{equation}
where extension to inhomogeneous elements is via linearity. 
\begin{definition}
Let $\mathcal{A}$ be a $G$-graded algebra and let $\rho: G \times G \rightarrow \mathbb{K}$ be a fixed commutation factor. Then $\mathcal{A}$ is said to be an \emph{almost commutative algebra} or a \emph{$\rho$-commutative algebra}, if $[f,g] =0$ for all $f,g \in \mathcal{A}$.
\end{definition}
The ethos is to treat elements of a $\rho$-commutative algebra as if they are global functions on some  `almost commutative graded manifold'.  Via some abuse of nomenclature, we will refer to elements of a  $\rho$-commutative algebra as \emph{functions}, and hence our already suggestive notation.
\begin{example}
Let $M$ be a smooth manifold, then $C^\infty(M)$ is a $\rho$-commutative algebra where $G = \{ e\}$ is the trivial group and the commutation factor is $\rho(e,e)=1$.     
\end{example}
\begin{example}
Let $\mathcal{M} = (M, \cO_{\mathcal{M}})$ be a supermanifold, then the algebra of global sections $C^\infty(\mathcal{M}) := \cO_M(|M|)$ is a $\rho$-commutative algebra where $G = \Z_2$, and the commutativity factor is a plus or   minus sign, i.e., $fg = (-1)^{|f| \, |g|}\, gf$ for homogeneous sections. 
\end{example}
The construction of the $\rho$-commutator leads to the following general definition.
\begin{definition}
Fix a discrete abelian group $G$ and a commutation factor $\rho: G \times G \rightarrow \mathbb{K}$. A \emph{$\rho$-Lie algebra} is a $G$-graded vector space $\mathfrak{g} = \bigoplus_{a \in G} \,  \mathfrak{g}_a$, equipped with a $\mathbb{K}$ bilinear map  $[-,-] : \mathfrak{g} \times \mathfrak{g} \rightarrow \mathfrak{g}$, that satisfies the following:
\begin{enumerate}[itemsep=0.5em]
    \item $[\mathfrak{g}_a, \mathfrak{g}_b] \subset \mathfrak{g}_{a+b}$;\label{eqn:LieA}
    \item $[u,v] = - \rho(|u|, |v|) \, [v,u]$; and \label{eqn:LieB}
    \item $[u, [v,w]] = [[u,v], w] + \rho(|u|, |v|)\, [v, [u,w]]$, \label{eqn:LieC}
\end{enumerate}
for all homogeneous $u,v,w \in \mathfrak{g}$. Extension to inhomogeneous elements is via linearity.
\end{definition}
\begin{remark}
$\rho$-Lie algebras also appear under the names generalised Lie algebras and colour Lie algebras.
\end{remark}
The condition \eqref{eqn:LieA} states that the Lie bracket carries zero $G$-degree, and condition \eqref{eqn:LieB} is understood as $\rho$-skewsymmetry. The final condition \eqref{eqn:LieC} is the \emph{$\rho$-Jacobi identity}, here written in Loday--Leibniz form. Note that the $\rho$-Jacobi identity is the classical Jacobi identity up to the twisting by the commutation factor.
\begin{example}
The $\rho$-commutator \eqref{eqn:RhoCom} can be shown to give a $\rho$-Lie algebra.
\end{example}
\begin{example}
Lie algebras and Lie superalgebras are examples of $\rho$-Lie algebras.
\end{example}
We will be interested in a specialised notion of a derivation of a $\rho$-commutative algebra that is compatible with the $G$-grading and the commutation factor. The reader can keep in mind even and odd vector fields on a supermanifold as the primary examples.
\begin{definition}
Let $\mathcal{A}$ be a $G$-graded algebra (not necessarily $\rho$-commutative). A \emph{$\rho$-derivation} $X$ of $\mathcal{A}$ of $G$-degree $|X| \in G$ is a linear map $X : \mathcal{A}  \rightarrow \mathcal{A}$ that satisfies the $\rho$-derivation rule
$$X(fg)  =  X(f)g  + \rho(|X| , |g|) \, f X(g)\,,$$
for all (homogeneous) $f,g \in \mathcal{A}$. The $G$-graded vector space of $\rho$-derivations of $\mathcal{A}$ we denote as $\rho \Der(\mathcal{A})$.
\end{definition}
Extension to inhomogeneous $\rho$-derivations is via linearity. The $\rho$-derivation rule is the twisted version of the standard derivation rule due to the presence of the commutation factor. Derivations, so linear maps that satisfy the standard derivation rule, can be defined; however, $\rho$-derivations will play a central role in theory.  Note that the $\rho$-commutator can be naturally extended to linear maps such as $\rho$-derivations. Moreover, the reader can verify that the $\rho$-commutator of two $\rho$-derivations is again a $\rho$-derivation, and that $\rho \Der(A)$ forms a $\rho$-Lie algebra.\par 
If the $G$-graded algebra is $\rho$-commutative, then $\rho \Der(\mathcal{A})$ has a left $\mathcal{A}$-module structure given by
$(fX)(g) :=f \big( X(g)  \big)$. Note that $|fX| = |X| + |f|$. Moreover, as $\rho \Der(\mathcal{A})$ is both a $\rho$-Lie algebra and a left $\mathcal{A}$-module, there is the expected \emph{$\rho$-Leibniz rule}
\begin{equation}\label{eqn:RhoLeiDer}
[X, fY] =  X(f) Y + \rho(|X|, |f|)\, f [X,Y]\,,
\end{equation}
for all homogeneous $X \in \rho\Der(\mathcal{A})$, $Y \in\rho \Der(\mathcal{A}) $ and homogeneous $f \in \mathcal{A}$.\par 
The structure of $\rho \Der(\mathcal{A})$ for a $\rho$-commutative algebra, and especially \eqref{eqn:RhoLeiDer}, will direct our definition of $\rho$-Lie-Rinehart pairs as generalisations of graded Lie algebroids. We view $\rho$-derivations of a $\rho$-commutative algebra as vector fields on the `almost commutative manifold' whose algebra of `global functions' is $\mathcal{A}$.
\begin{example}[Vector fields on supermanifolds]
Let $\mathcal{M} = (M, \cO_{\mathcal{M}})$ be a supermanifold, then its homogeneous vector fields $ X \in \Vect(\mathcal{M})$ are $\rho$-derivations. 
\end{example}
The \emph{tensor product} of two $\rho$-commutative algebras $\mathcal{A}$ and $\mathcal{B}$ (with the same commutation factor) $\mathcal{A} \otimes_{\mathbb{K}} \mathcal{B}$, is defined as the standard tensor product of graded vector spaces, where the product is defined as
\begin{equation}
(f_1 \otimes \psi_1)(f_2 \otimes \psi_2) = \rho(|\psi_1|, |f_2|)\, (f_1 f_2 \otimes \psi_1 \psi_2)\,.
\end{equation}
The $\rho$-derivations of a tensor product of $\rho$-commutative algebras are of the from 
\begin{equation}\label{eqn:TenProdDer}
D = X\otimes L_{\psi} + L_f \otimes \mathcal{X} \in \rho\Der(\mathcal{A} \otimes_{\mathbb{K}} \mathcal{B}) \,,
\end{equation}
where $X \in \rho\Der(\mathcal{A})$, $\mathcal{X} \in \rho\Der(\mathcal{B})$, $f \in \mathcal{B}$ and $\psi \in \mathcal{A}$.  Note that $G \in |D| = |X| + |\psi|  = |\mathcal{X}| + |f|$.  Here $L$ is the left action of the algebra on itself. The module structure is given by 
$$(g \otimes \chi)D = \rho(|\chi|, |X|)\, gX \otimes L_{\chi \phi} + \rho(|\chi|, |f|)\, L_{gf} \otimes \chi \mathcal{X}\,.$$
%
%
\section{\texorpdfstring{$\rho$}{rho}-Lie-Rinehart Pairs and Carrollian Structures}\label{sec:LieRinPaiCarr}
\subsection{\texorpdfstring{$\rho$}{rho}-Lie-Rinehart Pairs}\label{subsec:rhoLieRinPai}
Following the definition of a Lie algebroid, we make the following algebraic definition of a  $\rho$-Lie-Rinehart pair in which the graded structures are compatible.  We remark that Ngakeu \cite[Definition 3.1]{Ngakeu:2017} first defined such structures in the context of Poisson almost commutative algebras.   
\begin{definition}\label{def:rhoLieRinPai}
Fix a discrete abelian group $G$ and a commutation factor $\rho : G \times G \rightarrow \mathbb{K}$. An \emph{almost commutative Lie-Rinehart pair} or a \emph{$\rho$-Lie-Rinehart pair}  is a pair $(\mathcal{A}, \mathfrak{g})$ consisting of a 
\begin{enumerate}[itemsep=0.5em]
    \item a $\rho$-commutative algebra $\mathcal{A}$; and
    \item a $\rho$-Lie algebra $\mathfrak{g}$,
\end{enumerate}
such that $\mathfrak{g}$ is a left $G$-graded $\mathcal{A}$-module, and there exists a left $\mathcal{A}$-module homomorphism 
$$\mathsf{a} : \mathfrak{g}\longrightarrow  \rho\Der(\mathcal{A})\,,$$
called the \emph{anchor}, that satisfies
\begin{subequations}
\begin{align}
& \mathsf{a}_{[u,v]} = [\mathsf{a}_u, \mathsf{a}_v]\,, \label{eqn:rhoLieHom}\\
& [u,f\,v] =  \mathsf{a}_u(f) \, v + \rho(|u|, |f|)\, f [u,v]\,,\label{eqn:rhoLeiRul}
\end{align}
\end{subequations}
for all $u, v \in \mathfrak{g}$ and $f \in \mathcal{A}$.
\end{definition}
\begin{remark}
Barreiro et al. \cite{Barreiro:2023} introduced the notion of  \emph{graded Lie-Rinehart algebras}, where the structure is given by a graded Lie algebra and a graded associative commutative algebra.  Their notion corresponds to ours when the commutation factor is the trivial commutation factor, i.e., $\rho(a,b) = 1$, for all elements $a,b \in G$. Super or $\Z_2$-graded Lie-Rinehart pairs were defined by Jadczyk \& Kastler \cite{Jadczyk:1988} with the nomenclature \emph{graded Lie-Cartan pairs}.
\end{remark}
Condition \eqref{eqn:rhoLieHom} states that the anchor map is a $\rho$-Lie algebra homomorphism, and condition \eqref{eqn:rhoLeiRul} we refer to as the \emph{$\rho$-Leibniz rule}. To keep in line with classical nomenclature, we will refer to elements of $\mathcal{A}$ as \emph{functions} and elements of $\mathfrak{g}$ as \emph{sections}.  
\begin{remark}
Although we have sections of a `noncommutative vector bundle' in mind, there is no requirement that the $\mathcal{A}$-module structure of $\mathfrak{g}$ be finitely generated and projective. 
\end{remark}
\begin{definition}\label{def:MorrhoLieRinPair}
 A \emph{morphism of $\rho$-Lie-Rinehart pairs} (over $\mathcal{A}$), $\phi : (\mathcal{A}, \mathfrak{g})\longrightarrow (\mathcal{A}, \mathfrak{g}')$ is a $G$-graded left $\mathcal{A}$-module homomorphism which is compatible with the $\rho$-Lie brackets and the anchors, i.e.,
\begin{enumerate}[itemsep=0.5em]
\item $|\phi(w)| = |w|$\,;
\item $\phi(f\, u +v) = f\, \phi(u) + \phi(v)$\,;
\item  $\phi\big([u,v]\big) =  [\phi(u), \phi(v)]'$\,;~\textnormal{and}
\item $\mathsf{a}_u = \mathsf{a}'_{\phi(u)}$\,,
\end{enumerate}
for all homogeneous $w \in \mathfrak{g}$, $u,v \in \mathfrak{g}$ and $f \in \mathcal{A}$. 
\end{definition}
\begin{example}[Lie algebroids]
 Lie algebroids and Lie superalgebroids provide a geometric class of $\rho$-Lie-Reinhart pairs where we identify $\mathfrak{g}$ as the Lie (super)algebra of global sections and $\mathcal{A}$ is the global algebra of functions on the base (super)manifold.  For the manifold case, the group is the trivial group, and the commutation factor is just the identity. For the supermanifold case, the group is $\Z_2$, and the commutation factor is a plus or minus sign.  
\end{example}
\begin{example}[Zero $\rho$-Lie-Rinehart pairs]  Any $\rho$-commutative algebra $\mathcal{A}$ has associated with it the $\rho$-Lie-Rinehart pair $(\mathcal{A}, \{ 0\})$, where $\{ 0\}$ is the zero module, and we equip this with the zero anchor and zero $\rho$-Lie bracket. This structure we refer to as the \emph{zero $\rho$-Lie-Rinehart pair of $\mathcal{A}$}.
\end{example}
\begin{example}[$\rho$-derivations]
Let $\mathcal{A}$ be a $\rho$-commutative algebra, then $(\mathcal{A}, \rho\Der(\mathcal{A}))$ is a $\rho$-Lie-Rinehart pair where the bracket is the standard $\rho$-commutator and the anchor is the identity map. Note that, in general, $\rho\Der(\mathcal{A})$ is not finitely generated and projective; nonetheless, $(\mathcal{A}, \rho\Der(\mathcal{A}))$ is considered the generalisation of a tangent Lie algebroid. 
\end{example}
\begin{example}[Foliated almost commutative algebras]\label{exa:FolAlmComAlg} Let $\mathcal{A}$ be a $\rho$-commutative algebra, together with a left  $\mathcal{A}$-submodule $\mathcal{F}_\rho \subset \rho\Der(\mathcal{A})$ such that $[\mathcal{F}_\rho , \mathcal{F}_\rho] \subset \mathcal{F}_\rho$. The pair $(\mathcal{A}, \mathcal{F}_\rho)$ is  referred to as a \emph{foliated almost commutative algebra} (see \cite[Definition 2.9]{Bruce:2020}). We thus have a $\rho$-Lie-Rinehart pair, where the anchor is the identity map and the $\rho$-Lie bracket is the standard $\rho$-Lie bracket of $\rho$-derivations.
\end{example}
\begin{example}[$\rho$-Lie algebras]
Any $\rho$-Lie algebra $\mathfrak{h}$ can be considered as a $\rho$-Lie-Rinehart pair by setting $\mathcal{A} = \mathbb{K}$ and defining the anchor to be the zero map $\mathsf{a} : \mathfrak{h} \rightarrow 0 \in \mathbb{K}$.
\end{example}
The kernel of the anchor of a $\rho$-Lie-algebroid is defined as standard 
$$\mathfrak{g}_{\mathcal{A}} = \ker(\mathsf{a}) := \big\{  u \in \mathfrak{g}~~|~~ \mathsf{a}_u = 0\big\}\,.$$
\begin{proposition}
Let $(\mathcal{A}, \mathfrak{g})$ be a $\rho$-Lie-Rinehart pair with anchor $\mathsf{a} : \mathfrak{g} \rightarrow \rho\Der(\mathcal{A})$. Then $(\mathcal{A},\mathfrak{g}_{\mathcal{A}}) $  is a  $\rho$-Lie-Rinehart pair whose anchor is the zero map.
\end{proposition}
\begin{proof}\
\begin{enumerate}[itemsep=0.5em]
    \item $\ker(\mathsf{a})\subset \mathfrak{g}$ is a left $\mathcal{A}$-submodule: $\mathsf{a}_{u+ fv} = \mathsf{a}_u + \mathsf{a}_{fv} = \mathsf{a}_u + f\,\mathsf{a}_{v} =0+0=0$ for all $u,v \in \ker(\mathsf{a})$ and $f \in \mathcal{A}$.
    \item  $\ker(\mathsf{a})\subset \mathfrak{g}$ is a $\rho$-Lie algebra: $\mathsf{a}_{[u,v]} = [\mathsf{a}_u, \mathsf{a}_v] = [0,0]=0$ for all $u,v \in \ker(\mathsf{a})$. Thus, $[u,v]$ is in the kernel. 
    \item The anchor is the zero map: $[u, f\, v] =  \mathsf{a}_{u}(f) + \rho(|u|, |f|)\, f \, [u,v] = \rho(|u|, |f|)\, f \, [u,v]$, for all $u,v \in \ker(\mathsf{a})$ and $f \in \mathcal{A}$.
\end{enumerate}    
\end{proof}
We refer to $(\mathcal{A}, \mathfrak{g}_{\mathcal{A}} )$ as the \emph{isotropy $\rho$-Lie-Rinehart pair} of $(\mathcal{A}, \mathfrak{g})$.
\begin{definition}
A $\rho$-Lie-Rinehart pair $(\mathcal{A}, \mathfrak{g})$ is said to be \emph{transitive} if the anchor map $\mathsf{a} : \mathfrak{g} \rightarrow \rho \Der(\mathcal{A})$ is surjective.
\end{definition}
In analogy with the classical case of Lie algebroids, associated with a transitive $\rho$-Lie-Rinehart pair is the short exact sequence of $\rho$-Lie-Rinehart pairs 
$$(\mathcal{A}, \{ 0\})\rightarrow (\mathcal{A}, \mathfrak{g}_{\mathcal{A}} ) \rightarrow (\mathcal{A}, \mathfrak{g}) \rightarrow (\mathcal{A}, \rho\Der(\mathcal{A})) \rightarrow (\mathcal{A}, \{ 0\})\,.$$
The definition of a metric on a $\rho$-Lie-Rinehart pair is essentially that of a bilinear form on $\mathfrak{g}$. We inform the reader that our definition of a metric will encompass degenerate metrics.    
\begin{definition}
 A \emph{metric}  on a $\rho$-Lie-Rinehart pair $(\mathcal{A}, \mathfrak{g})$ is a $\mathbb{K}$-bilinear map $\mathcal{G} : \mathfrak{g} \times \mathfrak{g} \rightarrow \mathcal{A}$ that satisfies the following:
 \begin{enumerate}[itemsep=0.5em]
     \item $|\mathcal{G}(u,v)| = |u| +|v|$;
     \item $\mathcal{G}(u,v) = \rho(|u|, |v|)\, \mathcal{G}(v,u)$; and 
     \item $\mathcal{G}(f u,v ) =  f \, \mathcal{G}(u,v)$,
 \end{enumerate}
 for all (homogeneous) $u,v \in \mathfrak{g}$ and $f \in \mathcal{A}$. The kernel of $\mathcal{G}$ is defined as standard, i.e., 
 $$\ker(\mathcal{G}) := \big\{  u \in \mathfrak{g}~~|~~ \mathcal{G}(u, -) = 0\big\}\,.$$
 A metric is said to be a \emph{non-degenerate metric} if $\ker(\mathcal{G})$ is trivial, i.e., contains only the zero section of $\mathfrak{g}$, and is said to be a \emph{degenerate metric} otherwise. If the $\rho$-Lie-Rinehart pair is $(\mathcal{A}, \rho\Der(\mathcal{A}))$, then we will speak of a metric on $\mathcal{A}$. 
\end{definition}
\begin{remark}
The notion of a (non-degenerate) metric on a $\rho$-commutative algebra $\mathcal{A}$ as a bilinear map $ \mathcal{G} : \rho\Der(\mathcal{A})\times \rho\Der(\mathcal{A}) \rightarrow \mathcal{A}$ goes back to at least Ngakeu \cite{Ngakeu:2007}. This closely parallels the classical notion of a Riemannian metric on a manifold and is particularly suited to almost commutative differential geometry. Note that the existence of a non-degenerate metric requires $\mathfrak{g} \stackrel{\sim}{\rightarrow} \mathfrak{g^*}$, with a chosen metric explicitly giving the isomorphism.  However, such an isomorphism might not exist. For example, a non-degenerate (even) metric on a supermanifold only exists when there is an even number of odd coordinates. Moreover, it is possible to consider metrics that carry non-zero $G$-degree, though we will not do so in this paper. 
\end{remark}
\begin{example}[Induced metric on the isotropy $\rho$-Lie-Rinehart pair] A metric $\mathcal{G}$ on any $\rho$-Lie-Rinehart pair $(\mathcal{A}, \mathfrak{g})$ incudes a metric $\mathcal{H}$ on $(\mathcal{A},\mathfrak{g}_{\mathcal{A}})$ given by 
$$\mathcal{H}(u,v) := \mathcal{G}(\iota(u), \iota(v))\,,$$
where $\iota: \mathfrak{g}_{\mathcal{A}} \hookrightarrow \mathfrak{g}$ is the canonical inclusion.
 \end{example}
\begin{example}[Tensor product of metrics]\label{exa:TensProdG}
Suppose we have two $\rho$-commutative algebras $\mathcal{A}$ and $\mathcal{B}$, equipped with metrics $\mathcal{G}_\mathcal{A}$ and $\mathcal{G}_\mathcal{B}$, respectively. The induced metric on the tensor product is given by 
$$\mathcal{G}_{\mathcal{A} \otimes_{\mathbb{K}}\mathcal{B}}(D_1, D_2) = \rho(|\psi_1|, |\mathcal|X_2|) \, \mathcal{G}_{\mathcal{A}}(X_1, X_2)\otimes \psi_1 \psi_2 +  \rho(|\mathcal{X}_1|, |f_2|)\, f_1 f_2 \otimes \mathcal{G}_{\mathcal{B}}(\mathcal{X}_1, \mathcal{X}_2)\,,$$
where 
$$D_i = X_i\otimes L_{\psi_i} + L_{f_i} \otimes \mathcal{X}_i \in \rho\Der(\mathcal{A} \otimes_{\mathbb{K}} \mathcal{B})\,,$$
$i = 1,2$.
\end{example}
The definition of a metric leads to the notion of a \emph{covariant $\rho$-tensor} of $G$-degree $|\mathcal{T}|$ and valency $p \in \mathbb{N}$ as a $\mathbb{K}$-multilinear map 
$$\mathcal{T} : \underbrace{\mathfrak{g}\times \cdots \times \mathfrak{g}}_{p-\textnormal{factors}}  \longrightarrow \mathcal{A}\,,$$
such that 
\begin{enumerate}[itemsep=0.5em]
    \item $| \mathcal{T}(u_1, \cdots , u_p)| =  |\mathcal{T}| + |u_1| + \cdots + |u_p|$;
    \item $\mathcal{T}(u_1, \cdots , f\, u_i, \cdots, u_p) =  \rho(|\mathcal{T}| + |u_1| + \cdots + |u_{i-1}|, |f|)\, f \, \mathcal{T}(u_1, \cdots , u_i, \cdots, u_p)$, where $i \in \{ 1,2,\cdots, p\}$,
\end{enumerate}
for all homogeneous $u_1, \cdots, u_p \in \mathfrak{g}$. Such maps are the direct generalisation of covariant Lie algebroid tensors to the almost commutative setting. Thus, a metric is a covariant $\rho$-tensor of $G$-degree zero, and valency $2$. For brevity, we will refer to $\rho$-tensors. \par 
The structure of a $\rho$-Lie-Rinehart pair allows for an algebraic definition of Lie derivatives with respect to sections.  Acting on functions we define $\mathcal{L}_u f := \mathsf{a}_u(f)$, and on sections we define $\mathcal{L}_u v := [u,v]$. Clearly, via the properties of the anchor map and $\rho$-Jacobi identity, in both these cases we have $[\mathcal{L}_u , \mathcal{L}_v]= \mathcal{L}_{[u,v]}$. Moreover, the $\rho$-Leibniz rule for the Lie bracket shows that 
$$\mathcal{L}_u (f\, v) =  (\mathcal{L}_u f)\, v + \rho(|u|, |f|)\, f \, (\mathcal{L}_u v)\,.$$
The Lie derivative can then be extended to $\rho$-tensors via the $\rho$-derivation rule. Specifically for metrics, we have
$$ \mathcal{L}_u \big(  \mathcal{G}(v,w)\big) =  \big( \mathcal{L}_u \mathcal{G} \big)(v,w) + \mathcal{G}(\mathcal{L}_uv,w) + \rho(|u|, |v|)\, \mathcal{G}(v, \mathcal{L}_u w)\,,$$
which leads to 
\begin{equation}
    \big( \mathcal{L}_u \mathcal{G} \big)(v,w) := \mathsf{a}_u(\mathcal{G}(v,w)) -  \mathcal{G}([u,v], w) -  \rho(|u|, |v|)\, \mathcal{G}(v, [u, w])\,,
\end{equation}
for all homogeneous $u,v,w \in \mathfrak{g}$. Extension to non-homogeneous sections is via linearity. The properties of the Lie derivative acting on functions and sections imply that $\mathcal{L}_{[u,v]}\mathcal{G} = [\mathcal{L}_u, \mathcal{L}_v]\mathcal{G}$.
\begin{definition}
  Let $\mathcal{G}$ be a metric on a   $\rho$-Lie-Rinehart pair $(\mathcal{A}, \mathfrak{g})$. Then a section $u \in \mathfrak{g}$ is said to be a \emph{Killing section} if 
  $$\mathcal{L}_u \mathcal{G} = 0\,.$$
\end{definition}
The Killing condition for $u \in \mathfrak{g}$ can be written as 
\begin{equation}
\mathsf{a}_u(\mathcal{G}(v,w)) =  \mathcal{G}([u,v], w) +  \rho(|u|, |v|)\, \mathcal{G}(v, [u, w])\,,
\end{equation}
for all $v,w \in \mathfrak{g}$. Clearly, the $G$-graded vector space of Killing sections forms a $\rho$-Lie subalgebra of $\mathfrak{g}$.
\subsection{Connections on \texorpdfstring{$\rho$}{rho}-Lie-Rinehart Pairs}\label{subsec:ConrhoLieRinPai}
Connections are ubiquitous throughout differential geometry and mathematical physics. Linear connections on bimodules over $\rho$-commutative  algebras can be traced back to Ciupală \cite{Ciupală:2003}; for earlier related concepts see  Dubois-Violette \&  Michor \cite{Dubois-Violette:1988}.  Lie algebroids provide enough structure to generalise the notion of a linear connection in which vector fields are replaced with sections of a vector bundle. Similarly, $\rho$-Lie-Rinehart pairs provide the precise structure needed to define connections that have tensorial curvature and torsion. 
\begin{definition}
A  \emph{$\rho$-Lie-Rinehart connection} on a $\rho$-Lie-Rinehart pair $(A, \mathfrak{g})$ is a $\mathbb{K}$-bilinear map $\nabla: \mathfrak{g} \times \mathfrak{g} \rightarrow \mathfrak{g}$, that satisfies the following:
\begin{enumerate}[itemsep=0.5em]
    \item $|\nabla_u v | = |u| + |v|$;
    \item $\nabla_{fu} v  =  f \, \nabla_u v$; and 
    \item $\nabla_u fv  = \mathsf{a}_u(f)\, v + \rho(|u|, |f|)\, f \, \nabla_u v$,
\end{enumerate}
for all homogeneous $u,v \in \mathfrak{g}$ and homogenous $f \in \mathcal{A}$. Extension to non-homogeneous elements is via linearity.
\end{definition}
For brevity, we will refer to \emph{$\rho$-connections}. The existence of $\rho$-connections is not, in general, guaranteed. This is in stark contrast to Lie algebroid connections (in the smooth real category).   In particular, we do not, in general, have the standard tools such as a partition of unity and coordinates.\par 
The curvature and torsion maps of a $\rho$-connection can be defined following the standard definitions, i.e.,
\begin{subequations}
\begin{align}
    & R_\nabla(u,v)w := \nabla_u \nabla_v w - \rho(|u|, |v|) \, \nabla_v \nabla_u w - \nabla_{[u,v]}w\,, \\
    & T_\nabla(u,v) := \nabla_u v - \rho(|u|, |v|) \, \nabla_v u  -  [u,v]\,,
\end{align}
\end{subequations}
respectively. Clearly, both the curvature and torsion are $\rho$-skewsymmetric, that is 
$$R_\nabla(u,v) w  = -\rho(|u|, |v|)\, R_\nabla(v,u) w\,, \qquad T_\nabla(u,v) =  - \rho(|u|, |v|)\, T_\nabla(v,u)\,.$$
\begin{proposition}
The curvature and torsion of a $\rho$-Lie-Rinehart connection are $\rho$-tensors.
\end{proposition}
\begin{proof} This follows from a minor modification of the standard proofs, taking care with the commutation factor (see \cite[Theorem 1]{Ciupală:2003}). For completeness, we include the calculations.\
\begin{description}
    \item[Curvature] For shorthand, we will simply denote the curvature by $R$.
    \begin{align*}
        R(u,fv)w & = \nabla_u \nabla_{fv} w - \rho(|u|, |v| + |f|)\, \nabla_{fv} \nabla_u w - \nabla_{[u, f v]}w\\
        &= \nabla_u (f \, \nabla_v w) - \rho(|u|, |v| + |f|) \, f \, \nabla_v \nabla_u w - \nabla_{\mathsf{a}_u(f)v}w - \rho(|u|, |f|)\, f \, \nabla_{[u,v]}w\\
        &= \mathsf{a}_u(f) \, \nabla_v w + \rho(|u|, |f|)\, f \, \nabla_u \nabla_v w - \rho(|u|, |v| + |f|)\, f\, \nabla_v \nabla_u w \\
        &- \mathsf{a}_u(f)\, \nabla_v w - \rho(|u|, |f|)\, f \, \nabla_[u,v]w\\
        &= \rho(|u|, |f|) \, f \, R(u,v)w\,.
    \end{align*}
    The $\rho$-skewsymmetry of the curvature establishes that $R(fu, v) = f R(u,v)$.  As for the third argument, we have:
    \begin{align*}
     &   R(u,v)f\, w = \nabla_u \nabla_{v} f\,w - \rho(|u|, |v|)\, \nabla_{v} \nabla_u f\,w - \nabla_{[u, v]} f\,w\\
        &= \mathsf{a}_u (\mathsf{a}_v(f))w+ \rho(|u|, |v|+|f|)\, \mathsf{a}_v(f) \nabla_u w + \rho(|v|, |f|)\, \mathsf{a}_u(f)\nabla_v w + \rho(|u|+ |v|, |f|)\, f \nabla_u \nabla_v w\\
        &- \rho(|u|, |v|)\, \mathsf{a}_v(\mathsf{a}_u(f)) w - \rho(|v|, |f|)\, \mathsf{a}_u(f) \nabla_v w - \rho(|u|, |v|+ |f|)\, \mathsf{a}_v(f) \nabla_u w\\
        &- \rho(|u|+ |v|, |f|)\rho(|u|, |v|)\, f \, \nabla_v \nabla_u w-  \mathsf{a}_{[u,v]}(f) \, w - \rho(|u|+ |v|, |f|)\, f \, \nabla_{[u,v]}w\,.
    \end{align*} 
    Cancelling terms and using $[\mathsf{a}_u,\mathsf{a}_v] = \mathsf{a}_{[u,v]}$, we obtain $R(u,v)f\, w =  \rho(|u|+ |v|, |f|)\, f \, R(u,v,)w$.\medskip 
    \item[Torsion] For shorthand, we will simply denote the torsion by $T$.
    \begin{align*}
        T(u,fv) &= \nabla_u fv - \rho(|u|, |v| + |f|)\, \nabla_{fv} u - [u, fv]\\
        &= \mathsf{a}_u(f)v + \rho(|u|, |f|)\, f \, \nabla_u v -  \rho(|u|, |v| + |f|) \, f\, \nabla_v u - \mathsf{a}_u(f) v -  \rho(|u|, |f|)\, f\, [u,v]\\
        & = \rho(|u|, |f|)\, f \,T(u,v)\,.
    \end{align*}
    The $\rho$-skewsymmetry of the torsion quickly establishes that $T(fu, v) = f \, T(u,v)$. Hence, the torsion is tensorial in nature. 
\end{description}
\end{proof}
We can extend $\rho$-connections to act on $\rho$-tensors using the $\rho$-derivation rule. We first define a $\rho$-connection acting on a function as a directional derivative using the anchor, i.e, $\nabla_u f := \mathsf{a}_u(f)$.  Then, using the metric as an example, we have
$$\nabla_u \big( \mathcal{G}(v,u)\big) =  \big( \nabla_u \mathcal{G}\big)(v,w) + \mathcal{G}(\nabla_u v, w) + \rho(|u|, |v|)\, \mathcal{G}(v, \nabla_u w)\,,$$
and so we define
\begin{equation}\label{eqn:CovDerMet}
  \big( \nabla_u \mathcal{G}\big)(v,w) := \mathsf{a}_u \big( \mathcal{G}(v,u)\big)-  \mathcal{G}(\nabla_u v, w) - \rho(|u|, |v|)\, \mathcal{G}(v, \nabla_u w)\,,
\end{equation}
for all $u,v,w \in \mathfrak{g}$.
\begin{definition}
Let $(\mathcal{A}, \mathfrak{g})$ be a $\rho$-Lie-Rinehart pair equipped with a metric $\mathcal{G}$. Then a $\rho$-connection $\nabla$ is said to be \emph{metric compatible} if $\nabla \mathcal{G} =0$.
\end{definition}
If a $\rho$-connection is both torsion-free and metric compatible, then via minor modification of the proof of  \cite[Theorem 3.6]{Ngakeu:2007}, we have the \emph{Koszul formula}
\begin{align}\label{eqn:KosForm}
    2 \mathcal{G}(\nabla_u v,w) &= \mathsf{a}_u \big( \mathcal{G}(v,w)\big) + \mathcal{G}([u,v],w)\\ \nonumber 
    &+  \rho(|u|, |v|+|w|)\, \left(\mathsf{a}_v \big( \mathcal{G}(w,u)\big) - \mathcal{G}([v,w],u)\right)\\ \nonumber 
    &-  \rho(|w|, |u|+|v|)\, \left(\mathsf{a}_w \big( \mathcal{G}(u,v)\big) - \mathcal{G}([w,u],v)\right)\,. 
\end{align}
Furthermore, \cite[Theorem 3.6]{Ngakeu:2007} (also see \cite[Theorem 1]{Bruce:2020b}) leads to the following.
\begin{theorem}
Let $(\mathcal{A}, \mathfrak{g})$ be a $\rho$-Lie-Rinehart pair equipped with a non-degenerate metric $\mathcal{G}$. Then there exists a unique torsion-free and metric compatible $\rho$-connection, referred to as the Levi-Civita $\rho$-connection.  
\end{theorem}
\subsection{Carrollian Structures on \texorpdfstring{$\rho$}{rho}-Lie-Rinehart Pairs}\label{subsec:CarrrhoLieRinPair} Generalising the notion of a Carrollian Lie algebroid (see \cite{Bruce:2025}), we have the following definition.
\begin{definition}\label{def:CarrhoLieRinPai}
A \emph{Carrollian $\rho$-Lie-Rinehart pair} is the quadruple $(\mathcal{A}, \mathfrak{g}, \mathcal{G}, \mathfrak{l})$, where $(\mathcal{A}, \mathfrak{g})$ is a $\rho$-Lie-Rinehart pair, $\mathfrak{l}$ is a free cyclic submodule of $\mathfrak{g}$ generated by a $G$-degree $0$ section, and $\mathcal{G}$ is a degenerate metric on the $\rho$-Lie-Rinehart pair such that $\ker(\mathcal{G}) = \mathfrak{l}$.\par 
A \emph{morphism of Carrollian $\rho$-Lie-Rinehart pairs} (over $\mathcal{A}$) is an isomorphism of $\rho$-Lie-Rinehart pairs (see Definition \ref{def:MorrhoLieRinPair}) $\phi : (\mathcal{A}, \mathfrak{g}) \rightarrow (\mathcal{A}, \mathfrak{g}')$, such that $\phi$ 
\begin{enumerate}[itemsep=0.5em]
\item is an isometry, i.e., $\mathcal{G}(u,v) =  \mathcal{G}'(\phi(u), \phi(v))$ for all $u,v \in \mathfrak{g}$; and
\item respects the kernels, i.e., $\phi(\mathfrak{l}) = \mathfrak{l}'$.
\end{enumerate}
\end{definition}
From the definition, $\mathfrak{l}$ is cyclic, meaning it  is generated by a section $\sigma \in \mathfrak{g_0}$, so every element $s \in \mathfrak{l}$ is of the form $s = f \, \sigma$, where $f \in \mathcal{A}$. Note that $\mathfrak{l} \subset \mathfrak{g}$ is $G$-graded and not purely of $G$-degree $0$.  Moreover, the generator of $\mathfrak{l}$ is far from unique.  In particular, as $\mathcal{A}$ is unital, $\mathbb{K}$ embeds naturally into $\mathcal{A}$ 
$$\mathbb{K}\hookrightarrow \mathcal{A}_0\,, \qquad k \mapsto k \, \Id_{\mathcal{A}}\,.$$
Functions of the form $f_0  = k \,\Id_{\mathcal{A}}$ are referred to as \emph{field elements}. If $k \in \mathbb{K}\setminus \{0\}$, then  such field elements are units, i.e.,  $f_0 = k\,\Id_{\mathcal{A}}$ has an inverse $f_0^{-1} = k^{-1} \, \Id_{\mathcal{A}}$. Thus, the algebra $\mathcal{A}$ has at least unit field elements. Then any other $\sigma' = f_0 \, \sigma$, where $f_0 \in \mathcal{A}_0$ is a unit will serve as a free generator for $\mathfrak{l}$.
\begin{proposition}
Let  $(\mathcal{A}, \mathfrak{g}, \mathcal{G}, \mathfrak{l})$, be a Carrollian $\rho$-Lie-Rinehart pair. Then $(\mathcal{A},\mathfrak{l})$ is a $\rho$-Lie-Rinehart pair.
\end{proposition}
\begin{proof}
We need to argue that $[\mathfrak{l}, \mathfrak{l}] \subset \mathfrak{l}$.  As $\mathfrak{l}$ is a free cyclic submodule with a generator $\sigma$ of $G$-degree $|\sigma| = 0$, we see that the $\rho$-skewsymmetry of the bracket implies that $[\sigma, \sigma] =  - [\sigma, \sigma]= 0$. A quick calculation then shows that, for any $s = f\, \sigma$ and $t = g \, \sigma \in \mathfrak{l}$,
$$[s,t] =  [f\sigma, g \sigma] = \big(f\, \mathsf{a}_\sigma(g) - \rho(|f|, |g|) \, g \, \mathsf{a}_\sigma(f)\big) \, \sigma \in \mathfrak{l}\,. $$
Restriction of anchor to $\mathfrak{l} \subset \mathfrak{g}$ is well-defined and satisfies the required $\rho$-Leibniz rule \eqref{eqn:rhoLeiRul}.
\end{proof}
The condition that $\mathfrak{l}$ is a free cyclic module is akin to the requirement that the kernel of the degenerate metric of a Carrollian Lie algebroid be a trivial line bundle. Or, in relation to the standard definition of a Carrollian manifold, $\mathfrak{l}$ has been chosen to mimic the existence of a nowhere vanishing Carroll vector field.  The structure of a Carrollian $\rho$-Lie-Rinehart pair defines a foliation of the algebra $\mathcal{A}$ via the anchor map; this is the algebraic analogue of the Carroll distribution/foliation of a Carrollian manifold. 
\begin{definition}
Let  $(\mathcal{A}, \mathfrak{g}, \mathcal{G}, \mathfrak{l})$ be a Carrollian $\rho$-Lie-Rinehart pair. The associated \emph{Carroll distribution} is left $\mathcal{A}$-module $\mathcal{C} := \mathsf{a}(\mathfrak{l})\subset \rho \Der(\mathcal{A})$.
\end{definition}
We have the direct analogue of the classical result that the Carroll distribution is involutive.
\begin{proposition}
   Let  $(\mathcal{A}, \mathfrak{g}, \mathcal{G}, \mathfrak{l})$ be a Carrollian $\rho$-Lie-Rinehart pair. The associated Carroll distribution $\mathcal{C} \subset \rho\Der(\mathcal{A})$ is involutive, i.e., $[\mathcal{C}, \mathcal{C}]\subset \mathcal{C}$.
\end{proposition}
\begin{proof}
As $\mathcal{C}$ is defined as the image of the anchor map restricted to $\mathfrak{l}$, all sections are of the form $\mathsf{a}_s$ for some (not necessarily unique) $s \in \mathfrak{l}$. Then as the anchor map is a homomorphism of $\rho$-Lie algebras, we observe that $[\mathsf{a}_s, \mathsf{a}_t] = \mathsf{a}_{[s,t]} \in \mathcal{C}$ for any and all $s,t \in \mathfrak{l}$
\end{proof}
In other words, $(\mathcal{A}, \mathcal{C})$ is a \emph{foliated almost commutative algebra} (see Example \ref{exa:FolAlmComAlg} and \cite[Definition 2.9]{Bruce:2020}). In the classical setting of Carrollian Lie algebroids, the notion of a Carroll distribution being singular is defined by the rank of the distribution changing point-wise. The classical Carroll distribution is of maximal rank $1$; in general, it may fluctuate between rank-$1$ and rank-$0$. In the current setting, where we don't have points, we need the following algebraic analogue. The \emph{primitive ideal of $\mathcal{C}$} is defined as
$$P_{\mathcal{A}}(\mathcal{C}) : = \big \{ f \in \mathcal{A} ~~|~~ f\, X = 0 ~\textnormal{for all} ~ X \in \mathcal{C} \big\}\,,$$
and captures the notion of a distribution being singular or non-singular (regular).
\begin{definition}
Let  $(\mathcal{A}, \mathfrak{g}, \mathcal{G}, \mathfrak{l})$ be a Carrollian $\rho$-Lie-Rinehart pair and $\mathcal{C}$ be its associated Carroll distribution. Then $\mathcal{C}$ is said to be a \emph{non-singular Carroll distribution}  if its primitive ideal is trivial, i.e., $P_{\mathcal{A}}(\mathcal{C}) = \{ 0 \}$, and is said to be a \emph{singular Carroll distribution} otherwise.
\end{definition}
Geometrically, we think of the Carroll distribution $\mathcal{C} \subset \rho\Der(\mathcal{A})$ as defining a ``null direction'' in $\mathcal{A}$;  this is the noncommutative analogue of the `absolute time' found in standard Carrollian geometry.  The interpretation of a non-singular Carroll distribution is that its generator $\mathsf{a}_\sigma$ is a free generator of the left $\mathcal{A}$-module $\mathcal{C}$.  This means that the Carroll distribution is a free cyclic submodule of $\rho\Der(\mathcal{A})$, and moreover, $\mathcal{C} \cong \mathcal{A}$. The reader should think of a non-singular Carroll distribution as the algebraic generalisation of the classical distribution spanned by a Carroll vector field on a Carrollian manifold, which is a nowhere vanishing vector field that defines the null direction.\par  
For singular Carroll distributions, we have a cyclic submodule, however $\mathsf{a}_\sigma$ is a not a free generator as we can always find a $f \in \mathcal{A}$ non-zero, such that $f \, \mathsf{a}_\sigma =0$ (or similarly for any $X \in \mathcal{C}$). Such Carroll distributions are the algebraic counterparts of a singular Stefan--Sussmann distribution generated by a singular Carroll vector field.  \par 
The dynamics of a Carrollian particle is described as follows. Given a  $\rho$-derivation $X \in \mathcal{C}$ of $G$-degree zero, we can define a flow on  $\mathcal{A}$ via a one-parameter group of $\mathbb{K}$-linear automorphisms:
$$\mathcal{A} \ni  f \mapsto f(t) := \rme^{t X}\, f\,,$$
with $t \in \R$. The exponential is understood as being a formal power series, and the condition on $G$-degree $X$ is essential for ensuring the flow preserves the graded structure of $\mathcal{A}$. This flow we interpret as ``Carrollian motion'' in $\mathcal{A}$, as it is determined by a $\rho$-derivation corresponding to the ``geometric null direction''.  This is, of course, analogous to the classical situation where a Carrollian particle is constrained to be static and only move in the null direction. \par 
An important class of Carrollian manifolds are referred to as stationary Carrollian manifolds; the degenerate metric is time-independent. This notion generalises directly to  Carrollian $\rho$-Lie-Rinehart pair.
\begin{definition}
 A Carrollian $\rho$-Lie-Rinehart pair $(\mathcal{A}, \mathfrak{g}, \mathcal{G}, \mathfrak{l})$ is said to be a \emph{stationary Carrollian $\rho$-Lie-Rinehart pair} if every element of $\mathfrak{l}$ is a Killing section.  
\end{definition}
We need only use a free generator of $\mathfrak{l}$ to test if a Carrollian structure is stationary. More formally, we have the following.\par 
\begin{proposition}
Let  $(\mathcal{A}, \mathfrak{g}, \mathcal{G}, \mathfrak{l})$ be a Carrollian $\rho$-Lie-Rinehart pair and let $\sigma \in \mathfrak{g}_0$ be the free generating element of $\mathfrak{l} \subset \mathfrak{g}$. If $\sigma$ is a Killing section, then $(\mathcal{A}, \mathfrak{g}, \mathcal{G}, \mathfrak{l})$ is a stationary Carrollian $\rho$-Lie-Rinehart pair. 
\end{proposition}
\begin{proof}
The Killing condition for the generator is $\mathsf{a}_\sigma \big(\mathcal{G}(v,w) \big) = \mathcal{G}([\sigma, v], \sigma) + \mathcal{G}(v, [\sigma, w])$.  Then for $\mathfrak{l} \ni s = f\, \sigma$ we have
\begin{align*}
    \mathsf{a}_s\big(\mathcal{G}(v,w)\big) &= f \, \mathsf{a}_\sigma \big ( \mathcal{G}(v,w)\big) =f \, \big(  \mathcal{G}([\sigma, v], w) + \mathcal{G}(v, [\sigma, w])\big)\\
    & =  \mathcal{G}([f\, \sigma, v], w) + \rho(|f|, |v| )\, \mathcal{G}(v, [f\, \sigma, w])\\
    &-  \rho(|f|, |v|)\, \mathcal{G}(\mathsf{a}_v(f)\, \sigma, v)  -  \rho(|f|, |v|+ |w|)\, \mathcal{G}(v, \mathsf{a}_w(f)\, \sigma) \\
    &= \mathcal{G}([s, v], w) + \rho(|s|, |v| )\, \mathcal{G}(v, [s, w])\,,
\end{align*}
where have used $\mathfrak{l}= \ker\big( \mathcal{G}\big)$.
\end{proof}
As $\mathfrak{l}$ is a submodule of $\mathfrak{g}$, the quotient module $\mathcal{E} := \mathfrak{g}/ \mathfrak{l}$ is well-defined. Via the standard constructions, we understand the quotient module as cosets, i.e., $\mathcal{E} \ni \mathsf{x} = u + \mathfrak{l}$. The addition if defined as $\mathsf{x} + \mathsf{y} = (u + \mathfrak{l}) + (v + \mathfrak{l}) =  u+v + \mathfrak{l}$, and the multiplication is defined as $f \, \mathsf{x} = f \, (u + \mathfrak{l}) =  f \, u + \mathfrak{l}$. However, in general $\mathcal{E}$ is not a $\rho$-Lie algebra as the $\rho$-Lie bracket will not descend to the quotient. For the $\rho$-Lie bracket on $\mathcal{E}$ to be well-defined, it must not depend on the choice of representatives of the coset. This implies that $\mathfrak{l}$ must be an ideal, i.e., $[s, u ]\in \mathfrak{l}$ for all $s \in \mathfrak{l}$ and $u \in \mathfrak{g}$. In general, the ideal condition will not hold. 
\begin{example}[Carrollian Lie algebroids and manifolds]
Carrollian Lie algebroids c.f. \cite{Bruce:2025} are examples of Carrollian $\rho$-Lie-Rinehart pairs with $\mathcal{A} = C^\infty(M)$, $\mathfrak{g} = \Sec(A)$ and  where $\pi :A \rightarrow M$ is a Lie algebroid. Carrollian manifolds are examples of Carrollian Lie algebroids with $A = \sT M$ and its standard Lie algebroid structure.
\end{example}
\begin{example}[The superdomain $\R^{2|2}$]
Consider the supermanifold $\R^{2|2} = \big(\R^2, C^\infty_{\R^2}(-)[\theta^1, \theta^2] \big )$, here $C^\infty_{\R^2}(-)$ is the sheaf of smooth functions on $\R^2$, and $\theta^1, \theta^2$ are formal variables subject to $\theta^\alpha \theta^\beta = - \theta^\beta \theta^\alpha$ ($\alpha, \beta \in \{1,2\}$). Note  that this implies $(\theta^\alpha)^2=0$.  We employ global coordinates $(x^i, \theta^{\alpha}) = (x,y, \theta^1, \theta^2)$, where $x^i$ form a global coordinate system on $\R^2$.  We assign a $\Z_2$-grading $|x^i| = 0$, $|\theta^\alpha| = 1$, and commutation factor $\rho(a,b) = (-1)^{ab}$.  The coordinate derivatives form a basis of the $\Z_2$-graded derivations (vector fields) on $\R^{2|2}$. We then define a degenerate metric via its action on the  basis as
$$\mathcal{G}(\partial_x,\partial_x) =0\,, \quad\mathcal{G}(\partial_y, \partial_y) = 1\,, \quad\mathcal{G}(\partial_{\theta^1}, \partial_{\theta^2}) = - \mathcal{G}(\partial_{\theta^2}, \partial_{\theta^1}) = 1\,,$$
where all other components are zero. Clearly, $\ker{\mathcal{G}}$ is generated by $\partial_x$, and so we have a simple example of a \emph{Carrollian supermanifold}. For details of supermanifolds, see, for example, \cite{Carmeli:2011}.
\end{example}
\begin{example}[The $\Z_2\times\Z_2$-domain $\R^{1|1,2,2}$]
We define the algebra $\mathcal{A} := C^{\infty}(\R)[[z, \zx^1, \zx^2, \theta^1, \theta^2]]$ of formal power series with smooth functions in one real variable as coefficients  in formal variables subject to 
\begin{align*}
    & z\zx^i = - \zx^i z\,, && z \theta^i = \theta^i z \,, && \zx^i \theta^j = \theta^j \zx^i\\
    & \zx^i \zx^j =  - \zx^j \zx^i\,, && \theta^i \theta^j= - \theta^j \theta^i\,, &&
\end{align*}
where $i, j \in \{1,2\}$. This is the algebra of global sections of the $\Z^2_2$-manifold $\R^{1|1,2,2}$ (see \cite{Covolo:2016}). Using coordinates $(x, z, \zx^i, \theta^j)$,  we assign $\Z_2^2 := \Z_2 \times \Z_2$ degrees 
$$|x|= (0,0)\,, \quad |z|= (1,1)\, , \quad |\zx^i| = (0,1)\, , \quad |\theta^j| = (1,0)\,,$$
together with the commutation factor 
$$\rho : \Z_2^2 \rightarrow \R^{\times}\,, \qquad \rho(a, b) = (-1)^{\langle a,b\rangle}\,,$$
where $\langle -, -\rangle$ is the standard inner product on $\Z_2^2$. The $\rho$-derivations, in more classical language, $\Z_2^2$-graded derivations, have free generators $\{\partial_x, \partial_z, \partial_{\zx^i}, \partial_{\theta^j} \}$, of the same $\Z_2^2$-degree as its associated coordinate.  We then consider $(\mathcal{A}, \rho\Der(\mathcal{A}))$ as a $\rho$-Lie-Rinehart pair under the standard $\Z_2^2$-graded commutator and taking the anchor to be the identity map. \par 
We define a metric $\mathcal{G}$ via its non-zero components $\mathcal{G}(\partial_z, \partial_z) =1$, $\mathcal{G}(\partial_{\zx^1}, \partial_{\zx^2}) = -\mathcal{G}(\partial_{\zx^2}, \partial_{\zx^1})$ and $\mathcal{G}(\partial_{\theta^1}, \partial_{\theta^2}) = -\mathcal{G}(\partial_{\theta^2}, \partial_{\theta^1})$. Clearly, $\ker(\mathcal{G}) = \Span_{\mathcal{A}} \{ \partial_x\}$ is spanned by a free generator of $\Z_2^2$-degree $(0,0)$. Thus, we have a Carrollian $\rho$-Lie-Rinehart pair.
\end{example}
\begin{remark}
For further details of  $\Z_2^n$-manifolds and Riemannian structures on them, the reader should consult \cite{Bruce:2020b}.
\end{remark}
\begin{example}[Tensor product construction]
Let $\mathcal{A}$ be an arbitrary $\rho$-commutative algebra equipped with a non-degenerate metric $\mathcal{G}_{\mathcal{A}}$. Assume $\mathcal{B}$ is also a $\rho$-commutative algebra  (with the same commutation factor as $\mathcal{A}$) that is concentrated in $G$-degree zero, and that $\rho\Der(\mathcal{B})$ is a cyclic module. We equip $\mathcal{B}$ with the zero metric.  Following Example \ref{exa:TensProdG}, we define on the tensor product $\mathcal{A}\otimes_{\mathbb{K}}\mathcal{B}$ 
$$\mathcal{G}(D_1, D_2) := \mathcal{G}_{\mathcal{A}\otimes_{\mathbb{K}}\mathcal{B}}(D_1, D_2) = \mathcal{G}_\mathcal{A}(X_1,X_2) \otimes \psi_1 \psi_2\,,$$
with $D_i = X_i \otimes L_{\psi_i} + L_{f_i} \otimes \mathcal{X}_i$, $i = 1,2$.  The kernel of $\mathcal{G}$ are $\rho$-derivations of the form $L_f \otimes \mathcal{X}$. As $\rho\Der(\mathcal{B})$ is cyclic, we have a (non-unique)  generator $\sigma$, and so
$$L_f \otimes \mathcal{X} = (f\otimes \psi)(\Id_{\mathcal{A}} \otimes \sigma)\,.$$
Note $\ker(\mathcal{G})$ has a single generator $(\Id_{\mathcal{A}} \otimes \sigma)$ and so, by definition, it is cyclic, and the generator is $G$-degree zero.  Thus, we have constructed a Carrollian $\rho$-Lie-Rinehart pair. \par 
As we have the canonical inclusion 
$$\rho\Der(\mathcal{B}) \hookrightarrow \rho \Der(\mathcal{A} \otimes_{\mathbb{K}} \mathcal{B})\,, \qquad \mathcal{X} \mapsto (\Id_{\mathcal{A}} \otimes \mathcal{X})\,,$$
the Carrollian distribution is $\mathcal{C} = \mathcal{A}\otimes \rho\Der(\mathcal{B})$.
\end{example}
We can define a metric on the left $\mathcal{A}$-module $\mathcal{E}$ as follows. Consider two sections  $\mathsf{x}, \mathsf{y} \in \mathcal{E}$ together with chosen homogeneous lifts $u,v \in \mathfrak{g}$, i.e., $\pi(u) = \mathsf{x}$ and $\pi(v) = \mathsf{y}$ where $\pi : \mathfrak{g} \rightarrow \mathfrak{g}/\mathfrak{l} = \mathcal{E}$ is the canonical projection. We define a metric as
\begin{equation}\label{eqn:MetOnE}
    \mathcal{G}_{\mathcal{E}}(\mathsf{x}, \mathsf{y}) := \mathcal{G}(u,v)\,.
\end{equation}
For this metric to be well-defined, it should not depend on the chosen lifts. With this in mind, consider homogeneous lifts defined by $u' =  u + s$ and $v' = v + t$, with $s,t \in \mathfrak{l}$.  Then we observe that
$$\mathcal{G}(u',v') =  \mathcal{G}(u + s, v+t)= \mathcal{G}(u,v) + \mathcal{G}(u,t) + \mathcal{G}(s,v) + \mathcal{G}(s,t) = \mathcal{G}(u,v)\,,$$
as $\ker(\mathcal{G}) = \mathfrak{l} \ni s,t$. Thus, the metric $\mathcal{G}_{\mathcal{E}}$ is well-defined.  Recall that if given an $\mathsf{x} \in \mathcal{E}$ such that $\mathcal{G}_{\mathcal{E}}(\mathsf{x}, \mathsf{y}) = 0$ for all $\mathsf{y} \in \mathcal{E}$, means that $\mathsf{x}$ must be the zero section, then the metric is non-degenerate.  Observe that $\mathcal{G}_{\mathcal{E}}(\mathsf{x}, \mathsf{y}) = 0$ implies $\mathcal{G}(u,v) = 0$ for all lifts. As $\ker(\mathcal{G}) = \mathfrak{l}$, it must be the case that $u \in \mathfrak{l}$. Thus, as $\mathsf{x}$ is the projection of $u$ to $\mathfrak{g}/\mathfrak{l}$, it must be the case that $\mathsf{x}$ is the zero section. In short, we have established the following.
\begin{proposition}
Let  $(\mathcal{A}, \mathfrak{g}, \mathcal{G}, \mathfrak{l})$ be a Carrollian $\rho$-Lie-Rinehart pair and let $\mathcal{E}:=  \mathfrak{g}/\mathfrak{l}$ be the quotient module. Then $\mathcal{G}_{\mathcal{E}}$ defined by  \eqref{eqn:MetOnE} is a non-degenerate metric on $\mathcal{E}$.
\end{proposition}
To develop physics on Carrollian manifolds, a compatible connection is often required. We remark that, although the notion generalises to  Carrollian $\rho$-Lie-Rinehart pairs, the existence of $\rho$-connections, compatible or not, is not, in general, guaranteed.  However, we expect that physically reasonable Carrollian $\rho$-Lie-Rinehart pairs will admit $\rho$-connections.
\begin{definition}
A $\rho$-Lie-Rinehart connection $\nabla$ on a Carrollian $\rho$-Lie-Rinehart pair $(\mathcal{A}, \mathfrak{g}, \mathcal{G}, \mathfrak{l})$ is said to be a \emph{Carroll $\rho$-connection} if it is metric-compatible, i.e., $\nabla\mathcal{G} = 0$ (see \eqref{eqn:CovDerMet}).     
\end{definition}
The condition for a $\rho$-connection to be metric-compatible can be written as
\begin{equation}\label{eqn:MetComCon}
\mathsf{a}_u \big( \mathcal{G}(v,w)\big) =  \mathcal{G}(\nabla_u v, w) + \rho(|u|, |v|)\, \mathcal{G}(v, \nabla_u w)\,,    
\end{equation}
for all homogeneous $u,v,w \in \mathcal{A}$.
\begin{remark}
As the metric in the definition of the Carrollian $\rho$-Lie-Rinehart pair is degenerate, the Koszul formula \eqref{eqn:KosForm} cannot directly be used to construct a Carroll $\rho$-connection.  
\end{remark}
A Carroll $\rho$-connection is metric compatible by definition, and furthermore respects $\mathfrak{l}$ in the following sense.
\begin{proposition}
If $\nabla$ is a Carroll $\rho$-connection on a Carrollian  $\rho$-Lie-Rinehart pair $(\mathcal{A}, \mathfrak{g}, \mathcal{G}, \mathfrak{l})$, then $\nabla$ restricts to $\mathfrak{l}$, i.e., $\nabla_u s \in \mathfrak{l}$ for all $u \in \mathfrak{g}$ and $s \in \mathfrak{l}$.
\end{proposition}
\begin{proof}
Using \eqref{eqn:MetComCon} and setting $v = s \in \mathfrak{l}$ shows that $\mathcal{G}(\nabla_u s, w) = 0$.   As $\ker(\mathcal{G}) = \mathfrak{l}$, it must be the case that $\nabla_u s \in \mathfrak{l}$ for all $u \in \mathfrak{g}$ and $s \in \mathfrak{l}$.
\end{proof}
\subsection{Toy Examples}\label{subsec:ToyExa} In this subsection, we explore two simple examples based on the extended Manin quantum plane and the noncommutative $2$-torus. As well as defining Carrollian structures, we explicitly provide examples of Carroll connections.  We remark that Manin's plane provides a simplified mathematical model for the noncommuting coordinates that characterize a 2d electron gas in a strong magnetic field.  The noncommutative $2$-torus plays a prominent role in noncommutative field theory, matrix theories, and  the dynamics of
some low-energy excitations of open strings in the presence of non-vanishing supergravity fields.   
\begin{example}[The Extended Manin Quantum Plane]
As a toy example, we consider a version of Manin's quantum plane. The algebra that defines the extended quantum plane is $\mathbb{K}_q[x^{\pm1}, y^{\pm1} ]$ of polynomials subject to 
$$xy - q\, y x =0\,, \qquad x^{-1}x = x x^{-1} = \Id\,, \qquad y^{-1}y = y y^{-1} =\Id\,,$$
($q \in \mathbb{K}^\times$) where the other commutation rules can be deduced. This algebra is $\Z^2 := \Z \times \Z$ graded by defining 
\begin{align*}
    &  |\Id| = (0,0)\,, &&|x| = (1,0)\, , & &|y| = (0,1)\,,\\
    & |x^{-1}| = (-1,0)\,, && |y^{-1}| =  (0,-1)\,. &&
\end{align*}
We then define  $\mathbb{K}^2_q := \mathbb{K}_q[x^{\pm1}, y^{\pm1} ]$ with $(\mathbb{K}^2_q)_{n,m} := \Span_{\mathcal{A}}\{x^n y^m\}$. Observe that
$$(x^ny^m)(x^{n'} y^{m'}) =  q^{nm'- mn'}\, (x^{n'}y^{m'})(x^n y^m)\,,$$
and so the commutation factor is $\rho((n,m), (n'm')) = q^{nm'- mn'}$.\par 
The coordinate derivations $\partial_x, \partial_y$ (which are defined algebraically) for a free generating set of $\rho\Der(\mathbb{K}^2_q)$. The grading is 
$$|\partial_x| = (-1,0)\,, \qquad |\partial_y| = (0,-1)\,.$$
We then interpret $(\mathbb{K}^2_q,\rho\Der(\mathbb{K}^2_q) )$ as a $\rho$-Lie-Rinehart pair with the bracket being the $\rho$-commutator and the anchor  is the identity map.   As $x$ and $y$ are units, we can choose two free generators of $\rho\Der(\mathbb{K}^2_q)$ of $\Z^2$-degree $(0,0)$
$$\delta_ x := x \partial_x\,, \qquad \delta_y := y \partial_y\,.$$
We then define a metric via its action on these free generators:
\begin{align*}
&\mathcal{G}(\delta_x, \delta_x) =\Id\, ,\\
& \mathcal{G}(\delta_x, \delta_y) = \mathcal{G}(\delta_y, \delta_x)=0\,,\\
& \mathcal{G}(\delta_y ,\delta_y) = 0\,.
\end{align*}
As there are no non-zero divisors, $\ker(\mathcal{G})  = \Span_{\mathcal{A}} \{ \delta_y \}$, and so the kernel is a free cyclic submodule generated by the $G$-degree $(0,0)$ $\rho$-derivation. Thus, we have a Carrollian $\rho$-Lie-Rinehart pair.  The Carroll distribution is identified with $\ker(\mathcal{G})$ and is non-singular.\par 
As an example of a Carroll $\rho$-connection, we could pick the trivial connection. We will choose a connection in which there is a non-zero component in the kernel of $\mathcal{G}$: 
$$\nabla_{\delta_x}\delta_x :=  \delta_y\,, \qquad \nabla_{\delta_x} \delta_y = \nabla_{\delta_y} \delta_x= \nabla_{\delta_y} \delta_y := 0\,.$$
The only non-trivial condition to check is
$$\delta_x \left(\mathcal{G}(\delta_x, \delta_x) \right) \stackrel{?}{=} \mathcal{G}(\nabla_{\delta_x}\delta_x, \delta_x) + \mathcal{G}(\nabla_{\delta_x} \delta_x, \delta_x)\,.$$
Evaluating the left-hand side: 
$$\delta_x \left(\mathcal{G}(\delta_x, \delta_x) \right) = x\partial_x (\Id) = 0\,.$$ \\
Evaluating the right-hand side: 
$$\mathcal{G}(\nabla_{\delta_x}\delta_x, \delta_x) + \mathcal{G}(\nabla_{\delta_x}\delta_x, \delta_x) = 2\mathcal{G}(\nabla_{\delta_x}\delta_x, \delta_x)  =  2\mathcal{G}(\delta_y, \delta_x)  = 0\,,$$
as $\delta_y \in \ker(\mathcal{G})$.  Thus, we have a Carroll $\rho$-connection. Clearly, this connection is torsion-free. As for the curvature, the $\rho$-skewsymmetry of the curvature $\rho$-tensor means we need only check the two components $R(\delta_x, \delta_y)\delta_x$ and $R(\delta_x, \delta_y)\delta_y$.  
\begin{align*}
&R(\delta_x, \delta_y)\delta_x =  \nabla_{\delta_x}(\nabla_{\delta_y} \delta_x) - \nabla_{\delta_y} (\nabla_{\delta_x} \delta_x) - \nabla_{[\delta_x, \delta_y]}\delta_x =  - \nabla_{\delta_y} \delta_y =0\,,\\
&R(\delta_x, \delta_y)\delta_y = \nabla_{\delta_x}(\nabla_{\delta_y} \delta_y)- \nabla_{\delta_y}(\nabla_{\delta_x} \delta_y) - \nabla_{[\delta_x, \delta_y]}\delta_y=0\,.
\end{align*}
Thus, the Carroll $\rho$-connection is flat. 
\end{example}
\begin{example}[Canonical action $\rho$-Lie-Rinehart pair of the Noncommutative $2$-torus]\label{exa:ActLieRin} The $C^*$-algebra of functions on the noncommutative $2$-torus is generated (over $\C$) by two elements $u,v$ subject to 
$$uv  - \rme^{2\pi \rmi \, \theta}\, vu =0\,,$$
with $\theta \in \R$. We assign a $\Z^2 = \Z \times \Z$  grading by setting $|\Id| = (0,0)$, $|u| =(1,0)$ and $|v| = (0,1)$. The appropriate commutation factor is
$$\rho((n,m), (n',m')) =  \rme^{2\pi \rmi \, \theta(nm' -m'n)}\,.$$
The algebra is defined  as 
$$\mathcal{A}_\theta =\displaystyle \bigoplus_{n,m \in \Z} (\mathcal{A}_\theta)_{n,m}\,,$$
with $(\mathcal{A}_\theta)_{n,m}$ is the one-dimensional vector space over $\C$ of monomials of the form $u^nv^m$, noting $u^{-1} = u^*$ and $v^{-1} = v^*$. The coordinate derivations are $\partial_u$ and $\partial_v$ of $\Z^2$-degrees $(-1,0)$ and $(0,-1)$, respectively.  As $u$ and $v$ are units, the $\rho$-derivations are spanned by two free generators of $\Z^2$-degree $(0,0)$ given by 
$$\delta_u := 2\pi \rmi \,u  \partial_u\,, \qquad \delta_v := 2\pi \rmi \, v \partial_v\,.$$
 Moreover, these $\rho$-derivations are the infinitesimal generators of the canonical action of the classical $2$-torus $\mathbb{T} = S^1 \times S^1$ on the noncommutative 2-torus. The Lie algebra of the classical $2$-torus is $\R^2$ with the trivial Lie bracket. The infinitesimal action is 
\begin{align*}
    \alpha &  :  \R^2 \longrightarrow \rho\Der(\mathcal{A}_\theta)\,,\\
    & (a,b) \mapsto a \, \delta_u + b \,  \delta_v\,.
\end{align*}
Given that the Lie algebra $\R^2$ is trivial and that the generators $\delta_u$ and $\delta_v$ are degree $(0,0)$ and commute, we have a Lie algebra homomorphism. The canonical action $\rho$-Lie-Rinehart pair is 
$$(\mathcal{A}_\theta, A_\theta \otimes_{\C}\R^2)\,.$$
Sections are identified with pairs of functions $f_u$ and $f_v$, and collectively  $\mathbf{f}:= (f_u, f_v)$.  Then the anchor map and $\rho$-Lie bracket are given by 
\begin{align*}
    & \mathsf{a}_{\mathbf{f}} :=   f_u \delta_u + f_v \delta_v\,,\\
  & [\mathbf{f}, \mathbf{g}] := \big(  [\mathbf{f}, \mathbf{g}]_u,[\mathbf{f}, \mathbf{g}]_v \big)\,,
\end{align*}
with
$$ [\mathbf{f}, \mathbf{g}]_i := \sum_{j = \{u,v\}} \big( f_j\delta_j(g_i) -  \rme^{2\pi \rmi \, \theta(nm' -m'n)} \, g_j \delta_j(f_i)\big)\,,$$
where $|\mathbf{f}| = (n,m)$ and  $|\mathbf{g}| = (n',m')$. Extension to inhomogeneous sections is via linearity. \par 
We define a metric on the canonical action $\rho$-Lie-Rinehart pair of the noncommutative $2$-torus as 
$$\mathcal{G}(\mathbf{f}, \mathbf{g}) := f_v g_v\,.$$
As there are no non-zero divisors in $\mathcal{A}_\theta$, $\ker(\mathcal{G}) =  \Span_{\mathcal{A}_\theta}\{ (1,0) \}$. We set $\mathfrak{l} := \{ \mathbf{f} = (f_u,0 )~~|~~f_u \in \mathcal{A} \}$, which is a free cyclic submodule of $\mathcal{A}_\theta\otimes_{\C}\R^2$. Thus, we have constructed a Carrollian $\rho$-Lie-Rinehart pair.  The Carroll distribution is $\mathcal{C} := \mathsf{a}(\mathfrak{l}) \ni f_u \delta_u$, which is a free cyclic submodule of $\rho\Der(\mathcal{A}_\theta)$ and so is non-singular. \par 
As an example of a Carroll connection, we choose the trivial connection 
$$\nabla_{\mathbf{f}}\mathbf{g} := (\mathsf{a}_\mathbf{f}(g_u), \mathsf{a}_\mathbf{f}(g_v) ) =: \mathsf{a}_{\mathbf{f}}(\mathbf{g})\,.$$
We need to check 
$$\mathsf{a}_{\mathbf{f}}\big( \mathcal{G}(\mathbf{g}, \mathbf{h}) \big) \stackrel{?}{=} \mathcal{G}(\nabla_{\mathbf{f}}\mathbf{g}, \mathbf{h}) + \rho(|\mathbf{f}|, |\mathbf{g}|)\, \mathcal{G}(\mathbf{g}, \nabla_{\mathbf{f}}\mathbf{h})\,.$$
Evaluating the left-hand side: 
$$\mathsf{a}_{\mathbf{f}}\big( \mathcal{G}(\mathbf{g}, \mathbf{h}) \big) = \mathsf{a}_{\mathbf{f}}(g_v h_v)=\mathsf{a}_{\mathbf{f}}(g_v)h_v + \rho(|\mathbf{f}|, |\mathbf{g}|)\, g_v \mathsf{a}_{\mathbf{f}}(h_v)\,.$$
Evaluating the right-hand side: 
$$\mathcal{G}(\nabla_{\mathbf{f}}\mathbf{g}, \mathbf{h}) +\rho(|\mathbf{f}|, |\mathbf{g}|)\, \mathcal{G}(\mathbf{g}, \nabla_{\mathbf{f}}\mathbf{h}) = \mathsf{a}_{\mathbf{f}}(g_v)h_v + \rho(|\mathbf{f}|, |\mathbf{g}|)\, g_v \mathsf{a}_{\mathbf{f}}(h_v)\,.$$
Thus, we do indeed have a Carroll connection. As we have a trivial connection, it is defined by the anchor alone; the curvature is zero. The connection is clearly torsion-free. 
\end{example}
\begin{remark}
Example \ref{exa:ActLieRin} is formally equivalent to the canonical action $\rho$-Lie-Rinehart pair of the quantum Euclidean group $E_q(2)$. Recall that $E_q(2)$ is described by generators $\bar{v}, v, \bar{t}, t$ subject to 
\begin{align*}
    & \bar{v} v = v \bar{v} = \Id\,, && t \bar{t} = q^2\,  \bar{t}t\,,\\
    & v t = q^2 \, t v\,, && \bar{v}t = q^{-2}\, t \bar{v}\,,
\end{align*}
$q \in \R^\times$. Other relations follow using $v^* = \bar{v}$ and $t^* = \bar{t}$. We can assign a $\Z \times \Z$ grading 
$$|\Id| = (0,0)\,, \quad |v| = (-1,1)\,, \quad |\bar{v}| = (1,-1)\,, \quad |t| =(0,1)\,, \quad |\bar{t}| = (1,0)\,,$$
and define the commutation factor as $\rho((n,m), (n',m')) := q^{-2(nm' - n'm)}$. As $v$ and $\bar{v}$ are units, we have a pair of degree $(0,0)$ $\rho$-derivations,
$$\delta_v := 2 \pi \rmi \, v \partial_v\, \qquad \delta_{\bar{v}} := 2 \pi \rmi \, \bar{v} \partial_{\bar{v}}\,,$$
which generate an infinitesimal action of the classical $2$-torus on $E_q(2)$. The associated $\rho$-Lie-Rinehart pair and along with the Carrollian structure, can be constructed via minor modification of the quantum $2$-torus constructions.
\end{remark}
%
%
\section{Concluding Remarks}
We have generalised Carrollian geometry to almost commutative geometry via $\rho$-Lie-Rinehart pairs, which are themselves a natural generalisation of Lie algebroids. It is shown that the core tenets of Carrollian geometry extend to the almost commutative setting. This opens up the study of noncommutative Carrollian geometry and physics via $\rho$-commutative algebras and $\rho$-Lie-Rinehart pairs.  The question of the existence of compatible $\rho$-connections has not been addressed. We do not expect that an arbitrary  Carrollian $\rho$-Rinehart pair admits a Carroll $\rho$-connection. This is in stark contrast to the classical case of Carrollian manifolds and Lie algebroids, where Carroll connections always exist (see \cite{Bruce:2025}). Thus, we expect that physically relevant examples of  Carrollian $\rho$-Lie-Rinehart pairs will admit Carroll $\rho$-connections while general examples do not. The construction of physically motivated examples is now paramount and is essential in addressing this question. \par 
 We briefly suggest possible future areas of exploration following on from this work.
 \medskip 

\noindent \textbf{Holography and Horizons:}  Given how Carrollian physics is expected to shed light on flat space holography (see \cite{Donnay:2022,Donnay:2023}, for example), it is natural to ask how noncommutativity modifies boundary symmetries and charges.  For instance, how does almost commutativity change the  Bondi-Metzner-Sachs (BMS) symmetry algebra at null infinity? Moreover, as horizons are prototypical examples of Carrollian manifolds, we expect that quantum horizons can be studied within the $\rho$-commutative setting. The $\rho$-Lie-Rinehart pair approach is intrinsic, suggesting that one should develop an ultra-relativistic limit within almost commutative Riemannian geometry, and compare it with the standard commutative case.  \\
\noindent \textbf{Other Approaches:} It is important to relate the $\rho$-commuting framework to complementary noncommutative Carrollian constructions, which remain underdeveloped. Specifically, Inönü--Wigner contractions of $\kappa$-deformed spacetimes and their symmetries, see \cite{Ballesteros:2020,Ballesteros:2020b,Ballesteros:2023,Bose:2025,Trześniewski:2024}, provide a distinct approach.  Establishing overlap or clear links between $\kappa$-deformations and $\rho$-commutativity is nontrivial, but potentially illuminating. Another direction is to seek a spectral triple-like approach (c.f. Connes \cite{Connes:1994}) to degenerate metrics and Carrollian structures.\\
\noindent\textbf{Condensed Matter Physics:} Carrollian physics has been applied to various aspects of condensed matter physics and fluid dynamics (see \cite[Part III]{Bagchi:2025}). Fractons, for instance, are an emergent quasiparticle with restricted mobility; they can only move in bound pairs, and such quasiparticles can be studied from the perspective of Carrollian physics (see \cite{Figueroa-O’Farrill:2023b}). This points to potential applications of noncommutative Carrollian geometry.
%
%
\section*{Acknowledgements}
The author thanks Péter Horváthy for his interest in this work and encouragement. Cordial thanks are extended to the anonymous referees for their comments on earlier drafts of this work. 
%
%

\end{document}